\documentclass[letterpaper, 10 pt, conference]{ieeeconf}  % Comment this line out if you need a4paper

\IEEEoverridecommandlockouts                              % This command is only needed if 
\usepackage{bm} % bold math

\usepackage{amsthm} % theorems, proofs etc.

\usepackage{mathtools}
\usepackage{amssymb}
\usepackage{commath}
\usepackage{amsmath}
\usepackage{cite}
\usepackage{xcolor} % Text color

\usepackage[colorlinks=true,linkcolor=black]{hyperref}
\usepackage{url}

\usepackage{balance}
\usepackage[font={small}]{caption}
\usepackage{subcaption}
\usepackage{graphicx}
\usepackage{siunitx}

\let\labelindent\relax
\usepackage{enumitem} % enumerate without indent

\newcommand{\N}{\mathbb{N}}

\newcommand{\R}{\mathbb{R}}

\newtheoremstyle{mystyle}%                % Name
  {}%                                     % Space above
  {}%                                     % Space below
  {}%                                     % Body font
  {}%                                     % Indent amount
  {\bfseries}%                            % Theorem head font
  {.}%                                    % Punctuation after theorem head
  { }%                                    % Space after theorem head, ' ', or \newline
  {\thmname{#1}\thmnumber{ #2}\thmnote{ (#3)}}%                                     % Theorem head spec (can be left empty, meaning `normal')

\theoremstyle{mystyle}

\newtheorem{theorem}{Theorem}
\newtheorem{lemma}{Lemma}
\newtheorem{remark}{Remark}

\DeclareMathOperator*{\argmin}{argmin} % no space, limits underneath in displays
\title{\LARGE \bf
Duality-based Convex Optimization for Real-time Obstacle Avoidance between Polytopes with Control Barrier Functions
}
\author{
Akshay Thirugnanam, Jun Zeng, and Koushil Sreenath
\thanks{This work was supported by National Science Foundation Grant CMMI-1931853.}
\thanks{All authors are  with Hybrid Robotics Group at the Department of Mechanical- Engineering, UC Berkeley, USA. \tt\small\{akshay\_t, zengjunsjtu, koushils\}@berkeley.edu}
\thanks{A descriptive video with animations can be found at \url{https://youtu.be/rpwre3IYPJE}.}
}

\begin{document}

\maketitle
\thispagestyle{empty}
\pagestyle{empty}

%%%%%%%%%%%%%%%%%%%%%%%%%%%%%%%%%%%%%%%%%%%%%%%%%%%%%%%%%%%%%%%%%%%%%%%%%%%%%%%%
%%%%%%%%%%%%%%%%%%%%%%%%%%%%%%%%%%%%%%%%%%%%%%%%%%%%%%%%%%%%%%%%%%%%%%%%%%%%%%%%
\begin{abstract}
Developing controllers for obstacle avoidance between polytopes is a challenging and necessary problem for navigation in tight spaces.
Traditional approaches can only formulate the obstacle avoidance problem as an offline optimization problem.
To address these challenges, we propose a duality-based safety-critical optimal control using nonsmooth control barrier functions for obstacle avoidance between polytopes, which can be solved in real-time with a QP-based optimization problem.
A dual optimization problem is introduced to represent the minimum distance between polytopes and the Lagrangian function for the dual form is applied to construct a control barrier function.
We validate the obstacle avoidance with the proposed dual formulation for L-shaped (sofa-shaped) controlled robot in a corridor environment.
We demonstrate real-time tight obstacle avoidance with non-conservative maneuvers on a moving sofa (piano) problem with nonlinear dynamics.
% To the best of our knowledge, this is the first time that real-time tight obstacle avoidance with non-conservative maneuvers is achieved on a moving sofa (piano) problem with nonlinear dynamics.
\end{abstract}
\section{Introduction}
\label{sec:introduction}
\subsection{Motivation}
Achieving safety-critical navigation for autonomous robots in an environment with obstacles is a vital problem in robotics research.
Recently, control barrier functions (CBFs) together with quadratic program (QP) based optimizations have become a popular method to design safety-critical controllers.
% Among existing approaches of implementation, the robots and other surrounding obstacles are usually approximated as points, ellipses or hyper-spheres, and the distance between these shapes is an explicit analytic expression that is chosen as a CBF.
% The obstacle avoidance performance under this approximation is usually conservative, as the shapes of the controlled robots or obstacles are usually over approximated.
% When robots or obstacles are approximated as polytopes, the obstacle avoidance maneuvers become less conservative, but the signed distance function between the robot and the obstacle is implicit and cannot be used directly as a CBF.
In this paper, we propose a novel duality-based approach to formulate the obstacle avoidance problem between polytopes into QPs in the continuous domain using CBFs, which could then be deployed in real-time.

\subsection{Related Work}
\subsubsection{Control Barrier Functions}
% background of CBFs
One approach to provide safety guarantees for obstacle avoidance in control problems is to draw inspiration from control barrier functions.
CBF-QPs~\cite{ames2019control} permit us to find the minimum deviation from a given feedback control input to guarantee safety.
The method of CBFs can also be generalized for high-order systems~\cite{nguyen2016exponential, xiao2019control}, discrete-time systems~\cite{agrawal2017discrete, zeng2021enhancing, almubarak2022safety} and input-bounded systems~\cite{zeng2021safety, agrawal2021safe, katriniok2021control, breeden2021high}.
It must be noted that early work on nonovershooting control in \cite{KrBe2006} could also have been used to obtain results similar to control barrier functions.
% CBFs for obstacle avoidance
Specifically, CBFs are widely used for obstacle avoidance~\cite{yue2019quintic, srinivasan2020synthesis, huang2020switched, marley2021synergistic, he2021rule, marley2021maneuvering} with a variety of applications for autonomous robots, including autonomous cars~\cite{chen2017obstacle}, aerial vehicles~\cite{wu2016safety} and legged robots~\cite{hsu2015control}.
The shapes of robots and obstacles are usually approximated as points~\cite{chen2017obstacle}, paraboloids~\cite{ferraguti2020control} or hyper-spheres~\cite{zeng2020safety}, where the distance function can be calculated explicitly as an analytic expression from their geometric configuration.
The distance functions for these shapes are differentiable and can be used as control barrier functions to construct a safety-critical optimal control problem.

However, these approximations usually over-estimate the dimensions of the robot and obstacles, e.g., a rectangle is approximated as the smallest circle that contains it.
When a tight-fitting obstacle avoidance motion is expected, as shown in Fig.~\ref{fig:snapshots-sofa-problem}, robots and obstacles are usually approximated as polytopes.
While this makes maneuvers less conservative for obstacle avoidance, computing the distance between two polytopes requires additional effort~\cite{gilbert1988fast}.
Moreover, since this distance is not in an explicit form, it cannot be used directly as a CBF.
Furthermore, the distance between polytopes is non-differentiable~\cite{adda2001non}, which necessitates the use of nonsmooth control barrier functions (NCBFs) to guarantee safety~\cite{glotfelter2017nonsmooth}, ~\cite{glotfelter2018boolean}.

\begin{figure}%[!htp]
    \centering
    \includegraphics[width=0.99\linewidth]{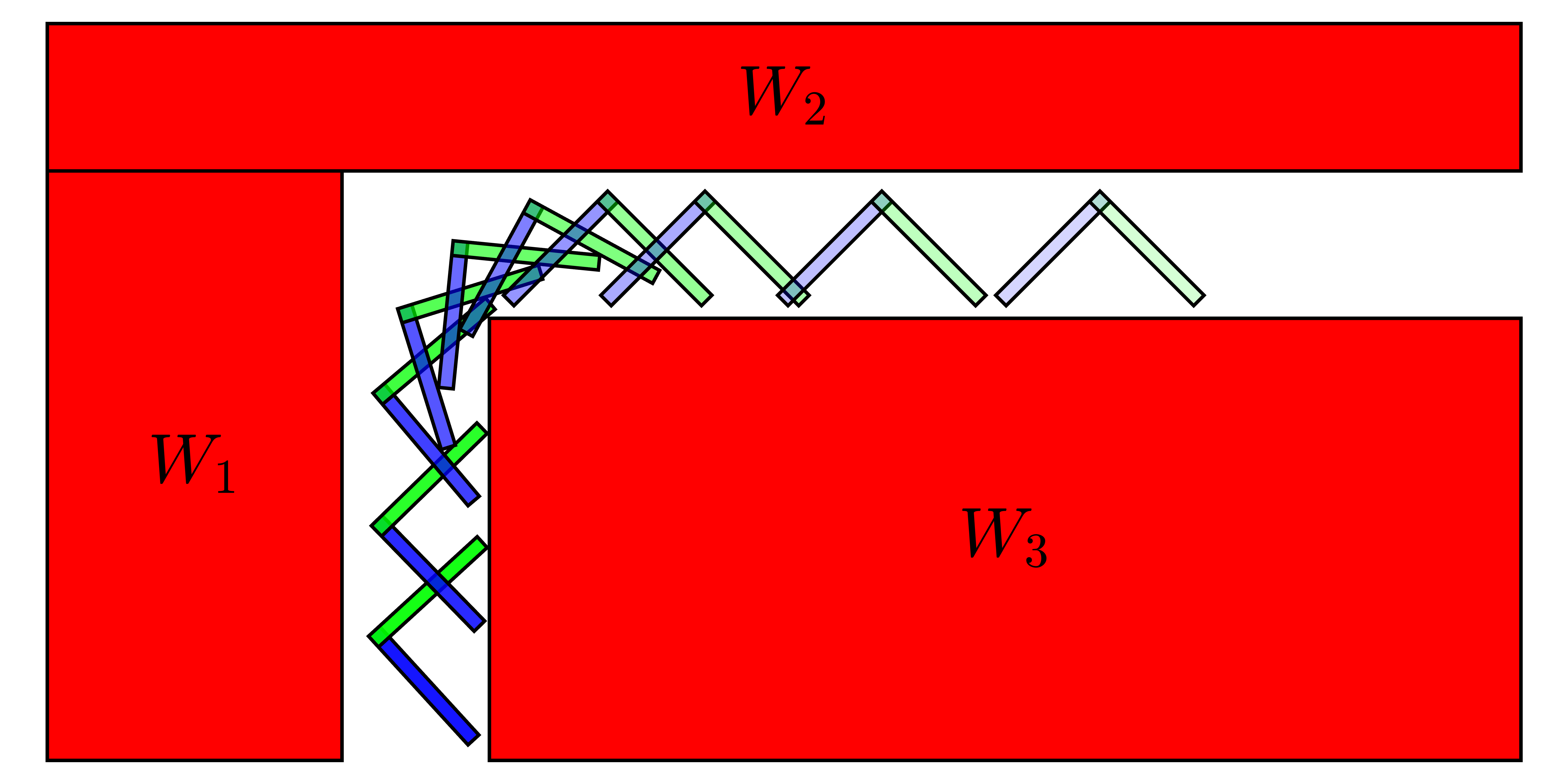}
    \caption{Snapshots of solving the moving sofa (piano) problem using our proposed formulation. It's shown that the controlled object can maneuver through a tight corridor whose width is smaller than the diagonal length of the controlled object, which cannot be achieved if we over-approximate these rectangle-shaped regions into spheres.}
    \label{fig:snapshots-sofa-problem}
    \vspace{-10pt}
\end{figure}

\subsubsection{Obstacle Avoidance between Polytopes}
We narrow our discussions about obstacle avoidance between polytopes into optimization-based approaches.
In~\cite{li2015unified}, obstacle avoidance between rectangle-shaped objects in an offline planning problem is studied, where collision avoidance is ensured by keeping all vertices of the controlled object outside the obstacle.
Generally, when controlled objects are polyhedral, the collision avoidance constraints can be reformulated with integer variables~\cite{grossmann2002review}.
This method applies well for linear systems using mixed-integer programming but cannot be deployed as real-time controllers for general nonlinear systems due to the complexity arising from integer variables.
The obstacle avoidance problem between convex regions could also be solved by using sequential programming~\cite{schulman2014motion}, where penalizing collisions with a hinge loss is considered through an offline optimization problem.

Recently, a duality-based approach~\cite{zhang2020optimization} was introduced to non-conservatively reformulate obstacle avoidance constraints as a set of smooth non-convex ones, which is validated on navigation problems in tight environments~\cite{zeng2020differential, shen2020collision, gilroy2021autonomous, firoozi2021formation}.
This idea does optimize the computational time compared with other ideas, but nonlinear non-convex programming is still involved.
Moreover, this approach can only be used for offline planning for nonlinear systems.
This philosophy has been extended into discrete-time control barrier functions (DCBFs) to enforce obstacle avoidance constraints between polytopes in real-time~\cite{thirugnanam2021safetycritical}.
However, the resulting DCBF formulation is still non-convex with non-convex DCBF or nonlinear system dynamics, and the computation time of the DCBF formulation might not be sufficiently low enough to be real-time.
On the other hand, a continuous-time formulation can result in a convex optimization formulation even for nonlinear systems, which leads to faster computation times.
Thus, continuous-time obstacle avoidance between polytopes requires a computationally efficient implementation, such as CBF-QPs and proper analysis on the nonsmooth nature of distance between polytopes to guarantee safety.
To summarize, real-time obstacle avoidance between polytopes with convex programming for general nonlinear systems is still a challenging problem.

\subsection{Contributions}

The contributions of this paper are as follows:

\begin{itemize}
    \item We propose a novel approach to reformulate a minimization problem for obstacle avoidance between polytopes for nonlinear affine systems into a duality-based quadratic program with CBFs.
    \item We establish the obstacle avoidance algorithm for the minimum distance between polytopes under a QP-based control law that guarantees safety, where the nonsmooth nature of the minimum distance between polytopes is resolved in the dual space.
    This formulation is used for real-time safety-critical obstacle avoidance.
    \item Our proposed algorithm demonstrates real-time obstacle avoidance at $\SI{50}{Hz}$ in the \textit{moving sofa (piano) problem}~\cite{howden1968sofa} with nonlinear dynamics, where an L-shaped controlled object can maneuver safely in a tight L-shaped corridor, whose width is less than the diagonal length of the controlled object.
\end{itemize}

% \subsection{Organization}
% The paper is structured as follows. Sec. \ref{sec:background} introduces the obstacle avoidance problem between two convex polytopes while reviewing the necessary background of discontinuous dynamical systems, minimum distance between polytopes, and nonsmooth control barrier functions (NCBFs).
% In Sec. \ref{sec:dual-formulation}, the duality-based formulation for obstacle avoidance using CBFs is presented.
% In Sec. \ref{sec:results}, the proposed approach is validated with numerical results on the moving sofa problem and concluding remarks are briefly discussed in Sec. \ref{sec:conclusion}.
% Due to limited space for submission, the appendix for proofs is excluded in this version and can be found in the full version of this paper~\cite{thirugnanam2021fastfullversion}.

\section{Background}
\label{sec:background}

% \subsection{Discontinuous Dynamical Systems}

We consider $N$ robots, with the $i$-th robot having states $x^i \in \mathcal{X} \subset \R^n$ and nonlinear, control affine dynamics:
\begin{equation} \label{eq:N-affine-systems}
\dot{x}^i(t) = f^i(x^i(t)) + g^i(x^i(t))u^i(t), \;\; i \in [N],
\end{equation}
where, $u^i(t) \in \mathcal{U} \subset \R^m$, $f^i: \mathcal{X} \rightarrow \R^n$ and $g^i: \mathcal{X} \rightarrow \R^{n\times m}$ are continuous, and $[N] = \{1,...,N\}$.
We assume $\mathcal{X}$ to be a connected set and $\mathcal{U}$ a convex, compact set.
While the dimensions of the states and inputs for each system can be different, we assume them to be the same across the systems for simplicity.
Throughout the paper, superscripts of variables denote the robot index and subscripts denote the row index of vectors or matrices.

\subsection{Closed-loop Trajectory for Discontinuous Inputs}
Polytopes have non-differentiable surfaces, and the minimum distance between polytopes could be non-differentiable at these points of non-differentiability.
Hence, enforcing safety constraints with the nonsmooth distance could result in the loss of continuity property of the feedback control.
Since $f^i$ and $g^i$ might not be Lipschitz continuous and $u^i(t)$ may be discontinuous, the solution of \eqref{eq:N-affine-systems} need not be unique or even well-defined.
In this case, to have a well-defined notion of a solution to \eqref{eq:N-affine-systems}, the dynamical system \eqref{eq:N-affine-systems} is turned into a differential inclusion.
Let $u^i: \mathcal{X} \rightarrow \mathcal{U}$ be some measurable feedback control law.
A valid solution for the closed loop trajectory is defined via the Filippov map \cite{glotfelter2018boolean} as
\begin{align} \label{eq:filippov-operator-def}
& \dot{x}^i(t) \in F[f^i+g^iu^i](x^i(t)) \\\nonumber
& := \text{co}\{\lim_{k\rightarrow \infty} (f^i{+}g^iu^i)(x_{(k)}): x_{(k)} \rightarrow x^i(t), x_{(k)} \notin \mathcal{Q}_f,\mathcal{Q}\}
\end{align}
where `co' stands for convex hull, $x_{(k)}$ denotes the $k$-th element of the sequence $\{x_{(k)}\}$, $\mathcal{Q}_f$ is a system-dependent zero-measure set, and $\mathcal{Q}$ is any zero-measure set.
The resulting map $F[f^i+g^iu^i]:\mathcal{X} \rightarrow 2^{\R^n}$ is a non-empty convex, compact, and upper semi-continuous set-valued map.
Here $2^{\R^n}$ denotes the power set of $\R^n$.
For a set-valued map $\Gamma: \mathcal{X} \rightarrow 2^{\R^n}$ to be upper semi-continuous at $a \in \mathcal{X}$, we require, $\forall \; \{a_{(m)}\} \in \mathcal{X}$ and $\{b_{(m)}\}$ such that $b_{(m)} \in \Gamma(a_{(m)})$, $\lim_{m\rightarrow \infty} a_{(m)} = a$ and $\lim_{m\rightarrow \infty} b_{(m)} = b \Rightarrow b \in \Gamma(a)$.
Then for all $x^i_0 \in \mathcal{X}$, a Filippov solution exists for the differential inclusion \eqref{eq:filippov-operator-def} with $x^i(0) = x^i_0$ \cite[Prop.~3]{cortes2008discontinuous}.
A Filippov solution on $[0,t]$ is an absolutely continuous map $x^i: [0,T] \rightarrow \mathcal{X}$ which satisfies \eqref{eq:filippov-operator-def} for almost all $t \in [0,T]$.
Throughout the paper, the term ``almost all" means for all but on a set of measure zero.

\subsection{Minimum Distance between Polytopes}
For robot $i$, we define the polytope $\mathcal{P}^i(x^i)$ as the $l$-dimensional physical domain associated to the robot at state $x^i \in \mathcal{X}$ with
\begin{equation} \label{eq:polytope-geometry}
\mathcal{P}^i(x^i) := \{z \in \R^l: A^i(x^i)z \leq b^i(x^i)\},
\end{equation}
where $A^i: \mathcal{X} \rightarrow \R^{r^i \times l}$ and $b^i: \mathcal{X} \rightarrow \R^{r^i \times 1}$ represent the half spaces that define the geometry of robot $i$ at some state, shown in Fig. \ref{fig:separating-plane-vector}.
We assume that the following properties hold for $A^i$ and $b^i$:
\begin{itemize}
    \item $\exists \; R, r > 0$ such that $\forall \; x^i \in \mathcal{X}, \; \exists \; c \in \R^l$ such that $B_r(c) \subset \mathcal{P}^i(x^i) \subset B_R(c)$, where $B_r(c)$ represents the open ball with radius $r$ centered at $c$.
    \item $A^i, b^i$ are continuously differentiable, and $\forall \; x \in \mathcal{X}$, the set of inequalities $A^i(x)z \leq b^i(x)$ does not contain any redundant inequality.
    \item $\forall \; x \in \mathcal{X}$, the set of active constraints at any vertex of $\mathcal{P}^i(x)$ are linearly independent.
\end{itemize}
The first assumption requires that the geometry of the robot $\mathcal{P}^i(x)$ be uniformly bounded with a non-empty interior for all $x \in \mathcal{X}$.
This assumption guarantees regularity conditions~\cite[Def.~2.3]{best1990stability} are met for minimum distance computations, which in turn guarantee the essential conditions of continuity~\cite[Thm.~2.3]{best1990stability} and differentiability~\cite[Thm.~2.4]{best1990stability} of minimum distance.
The last assumption requires that no more than $l$ half spaces intersect at any vertex.
As an example, a square pyramid in 3D does not satisfy this criterion.
Any polytope that does not satisfy this assumption can be tessellated into smaller polytopes, such as tetrahedra.

The square of the minimum distance between $\mathcal{P}^i(x^i)$ and $\mathcal{P}^j(x^j)$ is defined as $h^{ij}(x^i, x^j)$, where $h^{ij}: \mathcal{X}\times\mathcal{X} \rightarrow \R$ can be computed using the following QP:
\begin{equation}
\label{eq:min-dist-primal}
\begin{split}
    h^{ij}(x^i, x^j) := & \min_{\{z^i, z^j\}} \lVert z^i - z^j \rVert_2^2 \\
    \text{s.t.} \quad & A^i(x^i)z^i \leq b^i(x^i), \ A^j(x^j)z^j \leq b^j(x^j), \\
    &  z^i, z^j \in \R^l.
\end{split}
\end{equation}
Note that compared to the prior work on CBFs, 
the distance $h^{ij}$ is implicit and is a solution of a minimization problem.
By the regularity and smoothness assumptions on the polytopes, $h^{ij}$ is locally Lipschitz continuous~\cite[Lem.~1]{klatte1985stability}.
Variables $z^i, z^j \in \mathbb{R}^{l}$ denote points inside the polytopes $\mathcal{P}^i(x^i)$ and $\mathcal{P}^j(x^j)$ respectively.
Since the feasible sets of \eqref{eq:min-dist-primal} are non-empty (by assumption), convex, and compact, the solution to the QP \eqref{eq:min-dist-primal} always exists and is non-negative.
When the robots intersect with each other the minimum distance is uniformly zero, and there is no measure of penetration between the polytopes.

\subsection{Nonsmooth Control Barrier Functions}
For obstacle avoidance, we want to design a controller such that the minimum distance between any pair of robots $i$ and $j$ should be strictly greater than $0$.
We define a safe set of states $\mathcal{S}^{ij}$ as the  zero-superlevel set of the minimum distance between robots $i$ and $j$, $h^{ij}$ \cite{ames2019control},
\begin{equation} \label{eq:def-safe-set}
\mathcal{S}^{ij} := \{ (x^i,x^j): h^{ij}(x^i,x^j) > 0 \}^c,
\end{equation}
where $(\cdot)^c$ denotes closure of a set.
Note that since $h^{ij}(x^i,x^j) = 0$ for intersecting robot geometries, we use closure to obtain the correct safe set.
The closed-loop system is considered safe if $(x^i(t),x^j(t)) \in \mathcal{S}^{ij} \; \forall \; t \in [0, T]$, where $T$ is time till which the solution is defined.

Let $u^i(x^i), u^j(x^j) \in \mathcal{U}$ be feedback control laws that are measurable, with corresponding Filippov solutions $x^i(t), x^j(t)$ for $t \in [0,T]$.
Since $h^{ij}$ is locally Lipschitz and the state trajectories are absolutely continuous, $h^{ij}(t) := h^{ij}(x^i(t),x^j(t))$ is an absolutely continuous function, and is thus differentiable at almost all $t \in [0,T]$.
\begin{lemma}~\cite[Lem.~2]{glotfelter2017nonsmooth} \label{lem:nbf-def}
Let $\alpha:\R \rightarrow \R$ be a locally Lipschitz class-$\mathcal{K}$ function.
If
\begin{equation} \label{eq:nbf-constraint}
\dot{h}^{ij}(t) \geq -\alpha(h(t))
\end{equation}
for almost all $t \in [0,T]$ and $h(0) > 0$, then $h(t) > 0 \; \forall \; t \in [0,T]$, making the system safe.
\end{lemma}
In this case the absolutely continuous function $h^{ij}(t)$ is called a Nonsmooth Control Barrier Function (NCBF), which is a generalization of Control Barrier Functions (CBFs) to nonsmooth functions.
The constraint \eqref{eq:nbf-constraint} is called the NCBF constraint.
In the following section we derive a safety-critical feedback control law that satisfies this property.

\begin{remark}
For simplicity of discussion, the later analysis will be illustrated for a pair of robots $i$ and $j$, and results will be generalized where necessary. Further, for simplicity of notation, we denote $x := (x^i,x^j)$, $h(x) := h^{ij}(x^i,x^j)$, $\mathcal{S} := \mathcal{S}^{ij}$, and $u := (u^i,u^j)$ for the pair of systems $i$ and $j$.
\end{remark}
\section{NCBFs for Polytopes}
\label{sec:dual-formulation}

In this section, we will illustrate a general approach for imposing NCBF constraints for polytopes.
% \optional{A brief outline of the section is as follows: First we assume some feedback control law exists with its corresponding Filippov solution.
% Under this assumption, we derive our duality-based formulation which is sufficient to impose the required NCBF constraints.
% The duality-based formulation is a feedback control law, and we prove that all Filippov solutions of the feedback law satisfy the NCBF constraint, thus showing system safety.}
In order to enforce the NCBF constraints, we need to be able to write $\dot{h}(t)$ explicitly in terms of $\dot{A}^i(t), \dot{A}^j(t), \dot{b}^i(t), \dot{b}^j(t)$, which in turn depend on $u$.
In the following sub-section, we attempt to construct an explicit formula for $\dot{h}(t)$.

\subsection{Primal Approach}
To explicitly compute $\dot{h}(t)$, we first need to show that it exists.
% \optional{Since \eqref{eq:h-derivative-primal} is itself a minimization problem, using $\dot{h}(t)$ in the NCBF constraint in the implicit form will lead to large computation times.
% Instead we can find a constraint stronger than the NCBF constraint which can be written in explicit form for faster optimization performance.
% We use the dual formulation of \eqref{eq:min-dist-primal} to obtain the strong NCBF constraint.
% The reason for using the dual formulation instead of the primal is mentioned in Remark \ref{rem:duality-reason}.}
Let $x(t), t {\in} [0,T]$ be the Filippov solution corresponding to a feedback control law $u(x)$.
Let $\mathcal{O}(t)$ be the set of all optimal solutions to \eqref{eq:min-dist-primal} at time $t$.
For any pair of optimal solutions $(z^{*i}(t),z^{*j}(t)) \in \mathcal{O}(t)$, let $s^*(t) := z^{*i}(t)-z^{*j}(t)$ be the separating vector.
$s^*(t)$ is fundamental for the distance formulation between polytopes since it is related to the minimum distance $h^{ij}(t)=\lVert s^*(t)\rVert_2^2$.
We will next prove that $s^*(t)$ is unique, continuous, and right-differentiable.
\begin{lemma} \label{lem:unique-separating-vector}
For all $t {\in} [0,T]$, the separating vector $s^*(t)$ is unique for all pairs of primal optimal solutions $(z^{*i}(t),z^{*j}(t))$.
\end{lemma}
	
\begin{proof}
The feasible set of \eqref{eq:min-dist-primal} is convex and compact, and the cost function $\lVert z^{i}-z^{j}\rVert_2^2$ is convex.
Projecting the feasible set onto the $z^{i}+z^{j} = 0$ surface, the resulting set is also convex and compact.
Note that $s = z^{i}-z^{j} \in \R^l$ represents the projected co-ordinates on the $z^{i}+z^{j} = 0$ surface.
The resulting cost function is $\lVert s \rVert_2^2$, is strictly convex, and therefore, the optimal solution $s^*(t)$ of the projected problem exists and is unique.
Since for any pair of primal optimal solutions $(z^{*i}(t),z^{*j}(t))$, $z^{*i}(t)-z^{*j}(t)$ is an optimal solution to the projected problem, $z^{*i}(t)-z^{*j}(t)$ equals $s^*(t)$ for all pairs of optimal solutions of \eqref{eq:min-dist-primal}.
\end{proof}

For all $(z^{*i}(t),z^{*j}(t)) \in \mathcal{O}(t)$, we define $\text{Act}^i(t) \subset [r^i]$ ($\text{Act}^j(t) \subset [r^j]$) as the set of indices of constraints that are active for all $z^{*i}(t)$ ($z^{*j}(t)$) for $\mathcal{P}^i(t)$ ($\mathcal{P}^j(t)$).
Then,
\begin{align} \label{eq:optimal-affine-spaces}
\text{Aff}^i(t) & := \{z^i: A^i_{\text{Act}^i(t)}(t)z^i = b^i_{\text{Act}^i(t)}(t)\} \\\nonumber
\text{Aff}^j(t) & := \{z^j: A^j_{\text{Act}^j(t)}(t)z^j = b^j_{\text{Act}^j(t)}(t)\}
\end{align}
represent two parallel affine spaces such that the minimum distance between robots $i$ and $j$ at $t$ is the distance between $\text{Aff}^i(t)$ and $\text{Aff}^j(t)$.
These affine spaces can be points, hyperplanes, or even the entire space if the two polytopes intersect.
An example is pictorially depicted in Fig. \ref{fig:separating-plane-vector}.

We now assume that for almost all $t \in [0,T]$, $\exists \; \epsilon > 0$ such that the $\text{dim}(\text{Aff}^i(t))$, $\text{dim}(\text{Aff}^j(t))$, and dimension of the othogonal subspace common between $\text{Aff}^i(t)$ and $\text{Aff}^j(t)$ are constant for $\tau \in [t,t+\epsilon)$.
This is true when the set of times when states of the system oscillate infinitely fast has zero-measure.
In practice the states of the system do not oscillate infinitely fast due to limited control frequency and inertia of the system.
Then, under this assumption, we have the following result:
\begin{lemma} \label{lem:differentiable-separating-vector}
The separating vector $s^*(t)$ is continuous and right-differentiable for almost all $t \in [0,T]$.
\end{lemma}

\begin{proof}
Let $t$ be a time when the dimensions of the affine spaces and null space is constant for $\tau \in [t,t+\epsilon)$.
The separating vector $s^*(\tau)$ is the unique vector from $\text{Aff}^i(\tau)$ to $\text{Aff}^j(\tau)$ that is perpendicular to both of them.
These three constraints (that define $s^*(\tau)$ and establish orthogonality of $s^*(\tau)$ to $\text{Aff}^i(\tau)$ and $\text{Aff}^j(\tau)$) can be written in the form of a system of linear equalities with the matrix right-differentiable at $\tau=t$.
By the assumption on constant dimensions of the above spaces  ($\text{Aff}^i(t), \text{Aff}^j(t)$, and their common orthogonal subspace), this matrix has constant rank for $\tau \in [t,t+\epsilon)$.
By Lem.~\ref{lem:unique-separating-vector}, there is at least one solution to this system $\forall \; \tau \in [t,t+\epsilon)$.
A right-differentiable solution to $s^*(\tau)$ can be found from this linear system using Gauss elimination.
Since $s^*(t)$ is unique by Lem.~\ref{lem:unique-separating-vector}, $s^*(t)$ must be right-differentiable, and thus continuous, at $t$.
\end{proof}

To impose the NCBF constraint we then compute $\dot{h}(t) = \lim_{\delta \rightarrow 0^+} 1/\delta ( h(t+\delta) - h(t) )$ as:
\begin{equation} \label{eq:h-derivative-primal}
\dot{h}(t^+) = \frac{d}{d\tau} \lVert z^{*i}(\tau) - z^{*j}(\tau)\rVert_2^2 \biggl|_{\tau=t^+}. 
\end{equation}
% , and thus we need to evaluate the limit in \eqref{eq:h-derivative-primal} and obtain an explicit formula for $\dot{h}(t)$.
Although, $h(t) = \lVert s^*(t) \rVert^2_2$ is right-differentiable for almost all times, the primal optimal solutions $z^{*i}(t)$, $z^{*j}(t)$ could be non-differentiable or even discontinuous.
% \comments{The expression for $h(t)$, as calculated in \eqref{eq:min-dist-primal}, involves the primal variables $z^i$ and $z^j$.}
So, $\dot{h}(t)$ cannot be written explicitly as a minimization problem from \eqref{eq:h-derivative-primal}, since $\dot{z}^{*i}(t)$ and $\dot{z}^{*j}(t)$ need not be well-defined.
This leads us to consider the dual formulation instead.

As a motivation, consider enforcing the collision avoidance constraint $\mathcal{P}^i {\cap} \mathcal{P}^j {=} \emptyset$ explicitly.
Constraining the distance between any two points in $\mathcal{P}^i$ and $\mathcal{P}^j$ to be greater than $0$ is not sufficient, since they may not be the closest points.
This is due to the fact that \eqref{eq:min-dist-primal} is a minimization problem.
However, constraining a plane to separate $\mathcal{P}^i$ and $\mathcal{P}^j$ is sufficient to guarantee $\mathcal{P}^i {\cap} \mathcal{P}^j {=} \emptyset$ even if it is not the maximal separating plane.
The dual formulation of \eqref{eq:min-dist-primal} allows us to explicitly compute this separating plane constraint, which can be used in the NCBF constraint \eqref{eq:min-dist-primal}.
The reason for using the dual formulation instead of the primal is further elaborated upon later in Remark \ref{rem:duality-reason}.

\begin{figure}%[!htp]
    \centering
    \includegraphics[width=0.9\linewidth]{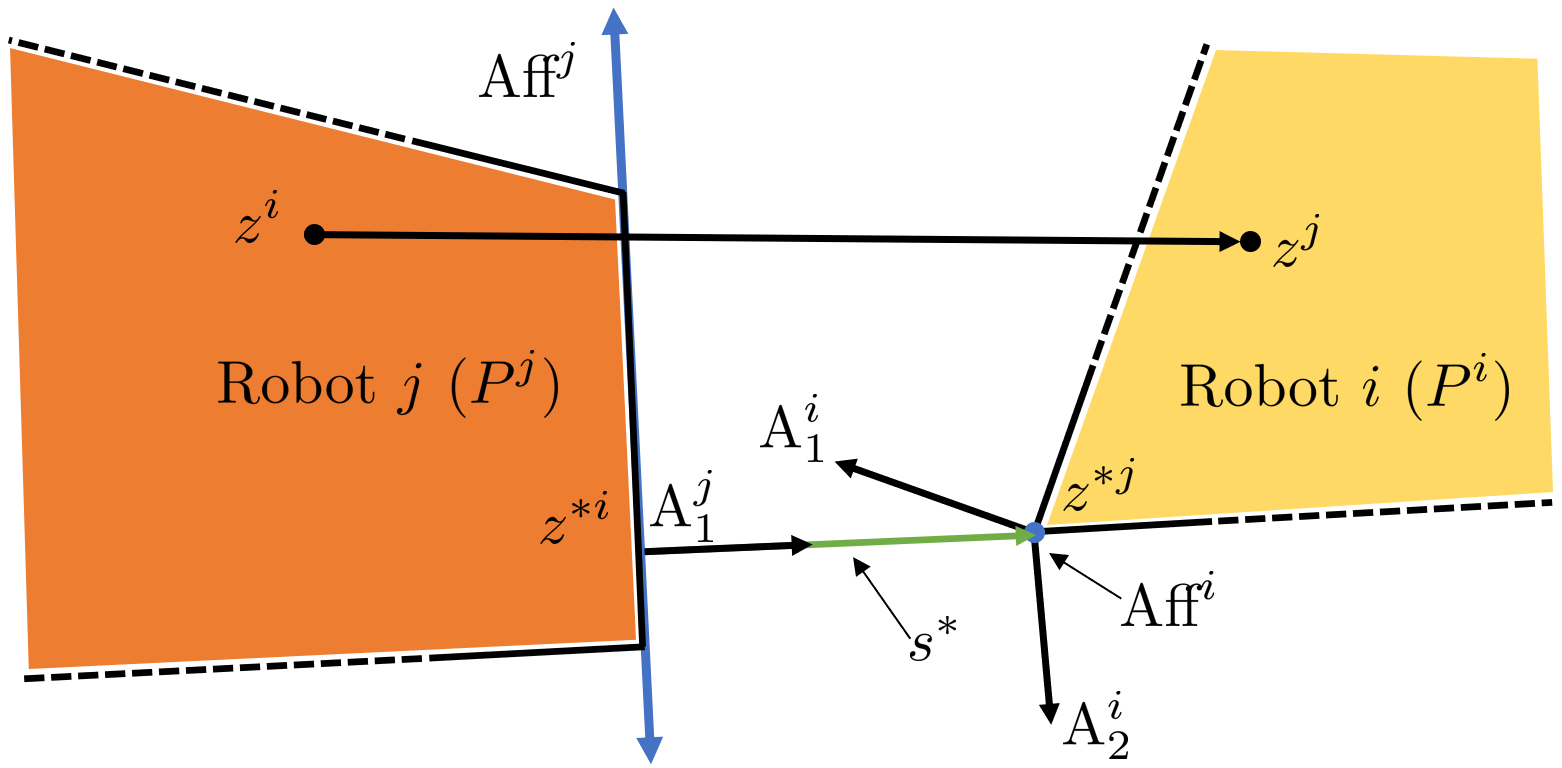}
    \caption{At any two configurations, the minimum distance between the robots $i$ and $j$ is the same as the minimum distance between the affine spaces Aff$^i$ and Aff$^j$, illustrated in blue.
    These spaces are affine extensions of some two faces of the robots.
    The points on robot $i$ and $j$ with the least distance are $z^{*i}$ and $z^{*j}$ respectively, and $s^*$ represents the vector with the smallest norm.
    The dynamics of the minimum distance between the robots is a hybrid system, with the discrete states being the pair of faces of the robots generating these affine spaces.
    For each discrete state, the minimum distance varies smoothly as the distance between the two affine spaces.}
    \label{fig:separating-plane-vector}
    \vspace{-10pt}
\end{figure}

\subsection{Dual Formulation} \label{subsec:ncbf-dual-formulation}

% Loosely speaking, since $h(\bm{s})$ in  \eqref{eq:min-dist-primal} is a minimization problem in the space variables $\bm{x}$ and $\bm{y}$, $\dot{h}(\bm{s},\bm{u})$ can be written as another minimization problem in $\dot{\bm{x}}$ and $\dot{\bm{y}}$. Since $\dot{h}(\bm{s},\bm{u})$ now is a solution of an optimization problem we cannot directly use it in a constraint to enforce the CBF constraint \eqref{eq:cbf-cons}. % since both $h(\bm{s},\bm{u})$ and the CBF constraint both depend on $\bm{u}$. 
% We could indirectly enforce the CBF constraint with some upper bound $\hat{\dot{h}}$ of $\dot{h}(\bm{s},\bm{u})$ obtained using a `feasible' $\dot{\bm{x}}$ and $\dot{\bm{y}}$. However, since $\dot{h}(\bm{s},\bm{u}) \leq \hat{\dot{h}}$, enforcing $\hat{\dot{h}} \geq -\alpha(h(\bm{s}))$ would not guarantee $\dot{h}(\bm{s},\bm{u}) \geq -\alpha(h(\bm{s}))$, and thus safety.

% The problem by using such an approach is that $\hat{\dot{h}}$ is an upper bound to $\dot{h}(\bm{s}, \bm{u})$. Instead, if we can find a lower bound of $\dot{h}(\bm{s},\bm{u})$, enforcing the CBF constraint on the lower bound would impose the CBF constraint \eqref{eq:cbf-cons} on $\dot{h}(\bm{s},\bm{u})$.

The dual program of a minimization problem is a maximization problem in terms of the corresponding dual variables.
For a quadratic optimization problem, as in \eqref{eq:min-dist-primal}, the dual program has the same optimal solution as that of \eqref{eq:min-dist-primal}.
So, differentiating the dual program will give $\dot{h}(t)$ as a maximization problem. Since any feasible solution to a maximization problem gives us a lower bound of the optimum, we use the dual problem of \eqref{eq:min-dist-primal} to get a lower bound of $\dot{h}(t)$. This enables us to express the NCBF constraint as a feasibility problem rather than an optimization problem.

To obtain the dual program, we first transform the constrained optimization problem \eqref{eq:min-dist-primal} to an unconstrained one by adding the constraints to the cost with weights $\lambda^i$ and $\lambda^j$, which are called the dual variables.
The unconstrained problem is optimized in terms of $z^i$ and $z^j$ to obtain the Lagrangian function $L(\lambda^i,\lambda^j)$ as
\begin{equation} \label{eq:lagrange-dual-func-expl}
L(\lambda^i, \lambda^j){=}{-}\frac{1}{4}\lambda^iA^i(x^i)A^i(x^i)^T\lambda^{i \; T}{-}\lambda^ib^i(x^i){-}\lambda^jb^j(x^j).
\end{equation}
The dual program is then defined as the maximization of the Lagrangian function.
\begin{lemma}
\label{lem:h-dual-qp}
The dual program corresponding to \eqref{eq:min-dist-primal} is:
\begin{equation}
\label{eq:min-dist-dual}
\begin{split}
    h(x) &= \max_{\{\lambda^i, \lambda^j\}} L(\lambda^i, \lambda^j) \\
    \text{s.t.} \quad & \lambda^iA^i(x^i) + \lambda^jA^j(x^j) = 0, \ \lambda^i, \lambda^j \geq 0.
\end{split}
\end{equation}
\end{lemma}

\begin{proof}
The proof is provided in Appendix \ref{app:proof-h-dual-qp}.
% The proof is provided in Appendix A in the full version of this paper~\cite{thirugnanam2021fastfullversion}.
\end{proof}

% \begin{remark} \label{rem:exist-dual-soln}
% For $\bm{\lambda_i}, \bm{\lambda_j}$ satisfying \eqref{eq:lagrange-inf-cond}, $\bm{\lambda_i}A_i(\bm{s_i})$ is the normal vector to a corresponding separating plane between the polytopes. This property is analyzed in~\cite{bertsimas1997introduction}.
% \end{remark}

Since an optimal solution to \eqref{eq:min-dist-primal} always exists, an optimal solution to \eqref{eq:min-dist-dual} also always exists.
Let $(z^{*i}(t),z^{*j}(t)) \in \mathcal{O}(t)$ and the active set of constraints at $z^{*i}(t)$ be $\text{Act}^i(z^{*i}(t),t)$.
Linear independence constraint qualification (LICQ) is said to be held at $z^{*i}(t)$ if $A^i_{\text{Act}^i(z^{*i}(t),t)}$ is full rank~\cite[Def.~2.1]{borrelli2017predictive}.
By the linear independence assumption on $\mathcal{P}^i(t)$, since the set of the active constraints at any vertex of $\mathcal{P}^i(t)$ are linearly independent, $A^i_{\text{Act}^i(z^{*i}(t),t)}$ is also full-rank for all $z^{*i}(t)$~\cite[Lem.~2.1]{borrelli2017predictive}.
Then, the primal problem \eqref{eq:min-dist-primal} is considered non-degenerate and there exists a unique dual optimal solution for \eqref{eq:min-dist-dual}~\cite[Lem.~2.2]{borrelli2017predictive}.
The dual optimal solution $(\lambda^{*i}(t),\lambda^{*j}(t))$ along with any primal optimal solution $(z^{*i}(t),z^{*j}(t))$ must then satisfy the KKT optimality conditions at $t$:
\begin{align}
2\lambda^{*i}(t)A^i(t) & = z^{*i}(t)^T - z^{*j}(t)^T = s^*(t)^T, \\\nonumber
2\lambda^{*j}(t)A^j(t) & = z^{*j}(t)^T - z^{*i}(t)^T = -s^*(t)^T, \\\nonumber
\lambda^{*i}_k(t) & = 0 \quad \text{for } k \notin \text{Act}^i(z^{*i}(t),t), \\\nonumber
\lambda^{*j}_k(t) & = 0 \quad \text{for } k \notin \text{Act}^j(z^{*j}(t),t).
\end{align}
Then, by LICQ, $A^i_{\text{Act}^i(t)}$ and $A^j_{\text{Act}^j(t)}$ have full rank and the non-zero components of the dual optimal solutions at $t$ can be written explicitly as:
\begin{align}
\lambda^{*i}(t){=}s^*(t)^T {A^{i \; \dagger}_{\text{Act}^i(t)}}, \
\lambda^{*j}(t){=}-s^*(t)^T {A^{j \; \dagger}_{\text{Act}^j(t)}},
\end{align}
where $(\cdot)^{\dagger}$ is the generalized inverse.
By assumption $\text{Aff}^i(\tau)$ and $\text{Aff}^j(\tau)$ have full rank for $\tau \in [t,t+\epsilon)$ for almost all $t \in [0,T]$, and thus $\lambda^{*i}(t)$ and $\lambda^{*j}(t)$ are right-differentiable at almost all $t \in [0,T]$.
Since $\lambda^{*i}(t), \lambda^{*j}(t), A^i(t)$ and $A^j(t)$ are right-differentiable, we can differentiate the constraints of \eqref{eq:min-dist-dual}.
Finally, we can explicitly write a linear program, which is obtained by differentiating the cost and constraints of \eqref{eq:min-dist-dual}, to calculate $\dot{h}(t)$ as:

\begin{lemma} \label{lem:deriv-h-dual-QP}
Let,
\begin{equation}
\label{eq:deriv-h-dual-QP}
\begin{split}
   g(t) = \max_{\{\dot{\lambda}^i,\dot{\lambda}^j\}} & \dot{L}(t,\lambda^{*i}(t),\lambda^{*j}(t),\dot{\lambda}^i,\dot{\lambda}^j) \\
    \text{s.t.} \quad & \dot{\lambda}^iA^i(t) + \lambda^{*i}(t)\dot{A}^i(t) \\
    & \hspace{10pt} + \dot{\lambda}^jA^j(t) + \lambda^{*j}(t)\dot{A}^j(t) = 0, \\
    & \dot{\lambda}^i_k  \geq 0 \quad \text{if } \lambda^{*i}(t)_k = 0, \\
    & \dot{\lambda}^j_k  \geq 0 \quad \text{if } \lambda^{*j}(t)_k = 0.
\end{split}
\end{equation}
where $\dot{L}(t,\lambda^i,\lambda^j,\dot{\lambda}^i,\dot{\lambda}^j)$ represents the time-derivative of Lagrangian function $L(\lambda^i, \lambda^j)$ and is as follows,
\begin{equation} 
\label{eq:def-deriv-L}
\begin{split}
    \dot{L} = & -\frac{1}{2}\lambda^iA^i(t)A^i(t)^T\dot{\lambda}^{i \; T} - \frac{1}{2}\lambda^iA^i(t)\dot{A}^i(t)^T\lambda^{i \; T} \\
& - \dot{\lambda}^ib^i(t) - \lambda^i\dot{b}^i(t) - \dot{\lambda}^jb^j(t) - \lambda^j\dot{b}^j(t).
\end{split}
\end{equation}
Then, for almost all $t \in [0,T]$, $\dot{h}^{ij}(t) = g(t)$.
\end{lemma}

\begin{proof}
Since $(\lambda^{*i}(t),\lambda^{*j}(t))$ is the optimal solution to \eqref{eq:min-dist-dual}, its derivative $(\dot{\lambda}^{*i}(t),\dot{\lambda}^{*j}(t))$ is a feasible solution to \eqref{eq:deriv-h-dual-QP} for almost all $t \in [0,T]$.
So, $\dot{h}^{ij}(t) = \frac{d}{dt}L(\lambda^{*i}(t),\lambda^{*j}(t)) \leq g(t)$ for almost all $t \in [0,T]$.
Let $(\dot{\lambda}^i,\dot{\lambda}^j)$ be a feasible solution to \eqref{eq:deriv-h-dual-QP}.
We can integrate the constraints of \eqref{eq:deriv-h-dual-QP} to find $(\bar{\lambda}^i(\tau),\bar{\lambda}^j(\tau)), \tau \in [t,t+\epsilon)$ which satisfy:
\begin{align}
(\bar{\lambda}^i(t),\bar{\lambda}^j(t)) = (\lambda^{*i}(t),\lambda^{*j}(t)) \\\nonumber
(\bar{\lambda}^i(t),\bar{\lambda}^j(t)) \text{ is dual feasible for \eqref{eq:min-dist-dual}}
\end{align}
So, $(\bar{\lambda}^i(\tau),\bar{\lambda}^j(\tau))$ are dual feasible and have cost less than $h^{ij}(\tau)$, i.e. $L(\bar{\lambda}(\tau),\bar{\lambda}(\tau)) {\leq} h^{ij}(\tau) \; \forall \tau{\in}[t,t{+}\epsilon)$ and $L(\bar{\lambda}(t),\bar{\lambda}(t)) {=} h^{ij}(t)$.
Differentiating the cost yields $\dot{h}^{ij}(t) \geq \dot{L}(t,\lambda^{*i}(t),\lambda^{*j}(t),\dot{\lambda}^i,\dot{\lambda}^j)$ for all feasible $(\dot{\lambda}^i,\dot{\lambda}^j)$, and thus $\dot{h}^{ij}(t) {\geq} g(t)$.
So, $g(t) {=} \dot{h}^{ij}(t)$ for almost all $t {\in} [0,T]$.
\end{proof}

% \aks{Add intuitive idea}

% \begin{rem}
% As long as $\mathcal{F}$ is non-empty, if no optimal dual solution $(\bm{\lambda_1}^*,\bm{\lambda_2}^*)$ lies in $\mathcal{F}$ then the result of Lem. \ref{lem:deriv-h-WDT} still holds because it only relies on the Weak Duality Theorem. This means that $\dot{L}(t,\bm{\lambda_1}(t),\bm{\lambda_2}(t),\dot{\bm{\lambda}}_{\bm{1}}(t),\dot{\bm{\lambda}}_{\bm{2}}(t))$ is always a lower bound to $\dot{h}(t)$. If the optimal dual solution is differentiable, we additionally have \eqref{eq:deriv-h-dual} to hold.
% So, non-differentiability of the optimal dual solution does not change the safety guarantee obtained from the CBF constraint. However, it could impact feasibility
% \end{rem}

Based on the linear program in \eqref{eq:deriv-h-dual-QP},  we can conservatively implement the NCBF constraint by enforcing, for some $(\dot{\lambda}^i,\dot{\lambda}^j)$ feasible for \eqref{eq:deriv-h-dual-QP},
\begin{equation} \label{eq:conserv-cbf-cons}
\dot{L}(t,\lambda^{*i}(t),\lambda^{*j}(t),\dot{\lambda}^i,\dot{\lambda}^j) \geq -\alpha(h(t)).
\end{equation}
Lem. \ref{lem:deriv-h-dual-QP} then guarantees that
\begin{equation} \label{eq:strong-nbf-proof}
\dot{h}(t) \geq \dot{L}(t,\lambda^{*i}(t),\lambda^{*j}(t),\dot{\lambda}^i,\dot{\lambda}^j) \geq -\alpha(h(t))
\end{equation}
which is the required NCBF constraint.

\begin{remark} \label{rem:duality-reason}
Due to the direction of inequality required for the NCBF constraint, $\dot{h}(t)$ needs to be expressed as a maximization problem.
This is the primary motivation for considering the dual problem, since writing $\dot{h}(t)$ using the primal problem results in a minimization problem \eqref{eq:h-derivative-primal}.
% Furthermore, although $s^*(t)$ is right-differentiable, $z^{*i}(t)$ and $z^{*j}(t)$ need not be right-differentiable or even continuous.
% So, $\dot{h}(t)$ cannot be written explicitly as a minimization problem from \eqref{eq:h-derivative-primal}, since $\dot{z}^{*i}(t)$ and $\dot{z}^{*j}(t)$ need not be well-defined.
\end{remark}

We use \eqref{eq:conserv-cbf-cons} to motivate a feedback control law to guarantee safety of the system.
The input $u$ implicitly affects $\dot{L}$ via the derivatives of the boundary matrices $A^i, A^j, b^i, b^j$.
Note that $\dot{L}$ is affine in $\dot{\lambda}^i$, $\dot{\lambda}^j$, and $u$.
So, \eqref{eq:conserv-cbf-cons} is a linear constraint in $\dot{\lambda}^i$, $\dot{\lambda}^j$, and $u$.
So, $ \forall \; x \in \mathcal{S}$, \eqref{eq:min-dist-dual} is used to compute $h(x)$, $\lambda^{*i}(x)$, and $\lambda^{*j}(x)$ and the optimal solution of following quadratic program is used as the feedback control:

% \noindent\rule{\columnwidth}{0.5pt}
% \textbf{Polytope-NCBF-QP:}
{\small
\begin{subequations}
\label{eq:dual-QP}
\begin{align}
   u^*(x) = & \argmin_{\{u,\dot{\lambda}^i,\dot{\lambda}^j\}} \lVert u - u^{nom}(x) \rVert^2_Q \\
    \text{s.t.} \quad &  \dot{L}(t,\lambda^{*i}(x),\lambda^{*j}(x),\dot{\lambda}^i,\dot{\lambda}^j,u)\geq-\alpha(h(x)-\epsilon_1^2) \label{subeq:dual-QP-ncbf} \\ %\nonumber
    % & \quad \geq  -\alpha(h(x)-\epsilon_1^2), \\
    & \dot{\lambda}^iA^i(x) {+} \lambda^{*i}(x)(\mathcal{L}_{f^i}A^i(x) {+} \mathcal{L}_{g^i}A^i(x)u) \label{subeq:dual-QP-plane} \\\nonumber
    & \quad {+} \lambda^{*j}(x)(\mathcal{L}_{f^j}A^j(x) {+} \mathcal{L}_{g^j}A^j(x)u)  = {-}\dot{\lambda}^jA^j(x) \\
%    & \quad = 0, \\
    & \dot{\lambda}^i_k  \geq 0 \ \text{if} \ \lambda^{*i}(x)_k{<}\epsilon_2, \ \dot{\lambda}^j_k  \geq 0 \ \text{if} \ \lambda^{*j}(x)_k{<}\epsilon_2, \label{subeq:dual-QP-non-neg}\\
    & |\dot{\lambda}^{i}| \leq M, |\dot{\lambda}^{j}| \leq M, \label{subeq:dual-QP-bounds} \\
    & u \in \mathcal{U},
\end{align}
\end{subequations}
}
% \noindent\rule{\columnwidth}{0.5pt}
where $u^{nom}(x)$ is a non-safe nominal feedback control law, $\mathcal{L}_{(\star)}(\cdot)$ represents the Lie derivative of $(\cdot)$ along $(\star)$, $Q \succ 0$ is the cost matrix, $M$ is a large number, and $\epsilon_1$ and $\epsilon_2 > 0$ are small constants.
$u^{nom}(x)$ can be obtained by a control Lyapunov function or by any tracking controller.
Note that \eqref{eq:dual-QP} is a feedback law which does not assume existence of Filippov solutions, since it is only a function of $x$ and not $t$.

\begin{remark} \label{rem:eps-1}
Since $h(x)$ is quadratic in nature, its gradient (if it exists) can be zero at $\partial \mathcal{S}$, which can affect forward invariance of $\mathcal{S}$~\cite[Rem.~5]{ames2019control}.
For example, in 1D, let $\dot{x} = u$ and $h(x) = x^2$.
Then at $x = 0$, the NCBF constraint reduces to $0\cdot u \geq 0$, which is true $\forall \; u \in \mathcal{U}$.
This would imply that the system remains safe irrespective of the input, which is not true.
This problem can be solved by setting $h(x) = x^2 - \epsilon_1^2$, which would result in non-zero gradient at $x = \epsilon_1$.
So, we consider an $\epsilon_1 > 0$ and the $\epsilon_1^2$-level set of $h$ as
\begin{equation}
\Omega_{\epsilon_1^2} = \{x: h(x) = \epsilon_1^2\}.
\end{equation}
We can then redefine $h(x)$ as $(h(x) - \epsilon_1^2)$.
If the new $h$ satisfies \eqref{eq:nbf-constraint}, then the $\epsilon_1^2$-superlevel set is safe.
% Thus, we can directly replace $h(x)$ with $(h(x)-\epsilon_1^2)$ in the controller derived.
\end{remark}

\begin{remark}
The constant $\epsilon_2$ is used to ensure upper semi-continuity of the feasible set of control inputs for \eqref{eq:dual-QP}.
This is similar to the almost-active gradient method used to prove safety in \cite[Thm.~3]{glotfelter2018boolean}.
% $\lambda^{*i}A^i$ is the separating vector $\mathcal{P}^i$ and $\mathcal{P}^j$.
% Geometrically, the last two constraints in \eqref{eq:dual-QP} mean that the separating plane cannot be more parallel to the facets of the polytopes than a margin proportional to $\epsilon_2$.
% Thus, a smaller value of $\epsilon_2$ will result in tighter control, but can result in collision if the control frequency is too low.
% The effect of increasing the value of $\epsilon_2$ is depicted in Fig. 
% Since upper semi-continuity of $h(t)$ is essential for safety, $\epsilon_2$ trades point-wise feasibility for persistent feasibility.
\end{remark}

We now prove that \eqref{eq:dual-QP} results in system safety.
\begin{theorem} \label{theorem:safety}
Let $x(0) \in \mathcal{S}$ and \eqref{eq:dual-QP} be feasible $\forall \; x \in \mathcal{S}$ for some locally Lipschitz class-$\mathcal{K}$ function $\alpha$. Then, using the feedback control law \eqref{eq:dual-QP}, the system remains safe irrespective of the cost function.
\end{theorem}

\begin{proof}
The proof is provided in Appendix \ref{app:safety-proof}.
% The proof is provided in Appendix B in the full version of this paper~\cite{thirugnanam2021fastfullversion}.
\end{proof}

\begin{remark} \label{rem:multiple-robot-obstacle}
This formulation can be extended to more than 2 robots by introducing a new pair of dual variables for each pair of robots, and the corresponding constraints in \eqref{eq:dual-QP}.
Note that the dual variable for robot $i$ corresponding to robot $j$ is different from that for robot $k$.
Thus, the analysis in Sec. \ref{sec:dual-formulation} remains valid.
Additionally, static or dynamic obstacles can be represented as uncontrolled robots and non-convex shaped robots can be represented through unions of different polytopes with each using the same states and inputs.
\end{remark}

\section{Numerical Examples}
\label{sec:results}
In this section, we consider the problem of moving an L-shaped sofa through an L-shaped corridor.
The approach presented in Sec. \ref{sec:dual-formulation} is applied to solve the problem in real-time.

%%%%%%%%%%%%%%%%%%%%%%%%%%%%%%%%%%%%%%%%%%%%%%%%%%%%%%%%%%%%%%%%

\subsection{Simulation Setup}
The sofa is the controlled object, with the side length of the L-shape as $\SI{1}{m}$ and the width as $\SI{0.1}{m}$, as illustrated in Fig. \ref{fig:snapshots-sofa-problem}. To implement our controller, we consider the sofa as the union of two $\SI{1}{m}\times\SI{0.1}{m}$ arms perpendicular to each other. The width of the corridor is $\SI{1}{m}$.

The states of the sofa are $x=(z_1,z_2,\theta) \in \R^2\times \R$, where $(z_1,z_2)$ is the position of the vertex at the intersection of the two arms, and $\theta$ is the angle of rotation. The inputs to the sofa are $(v,\omega)$, where $v$ is the speed of the sofa and $\omega$ the angular velocity. The velocity of the sofa is assumed to be along a $\frac{\pi}{4}$ angle to the arm. The nonlinear control affine dynamics of the sofa is then,
\begin{equation} \label{eq:ex-sofa-dyn}
\dot{z}_1 = v \cos(\theta + \frac{\pi}{4}), \ \dot{z}_2 = v \sin(\theta + \frac{\pi}{4}), \ \dot{\theta} = \omega.
\end{equation}
Further, we impose input bounds as $|v| \leq \SI{0.3}{m/s}$ and $|\omega| \leq \SI{0.2}{rad/s}$.
The corridor is represented as the union of three rectangular walls, shown in Fig. \ref{fig:snapshots-sofa-problem}. The polytopes corresponding to these walls are defined as follows:
\begin{equation}
\mathcal{W}^i \coloneqq \{z \in \R^2: A^{W^i}z \leq b^{W^i}\}, \; i \in [3]
\end{equation}
We note the general problem description in Sec. \ref{sec:background} allows for such static obstacles by eliminating the dependence of $A^{W^i}$ and $b^{W^i}$ on the state, as mentioned in Rem. \ref{rem:multiple-robot-obstacle}.

Similarly, the sofa is represented as:
\begin{equation}
\mathcal{P}^j(x) \coloneqq \{z \in \R^2: A^{S^j}(x)z \leq b^{S^j}(x)\}, \; j \in [2],
\end{equation}
where $\mathcal{P}^j(x)$ is denoted as Arm $j$.
Arm 1 is the blue-colored polytope in Fig. \ref{fig:snapshots-sofa-problem}, whereas Arm 2 is the green-colored one.
Again, our description allows for this by choosing the same states and inputs for the two polytopes.

The NCBF is chosen as the square of the minimum distance as in \eqref{eq:min-dist-primal}.
Notice that the QP-based program \eqref{eq:dual-QP} will always have a solution, as $v=0, \omega=0$ is always a feasible solution.
Since we have static obstacles and the controlled object consists of two polytopes, we only enforce NCBF constraints between $\mathcal{W}^i$ and $\mathcal{P}^j$ for $i \in [3], j\in [2]$.

A control Lyapunov function $V_{clf}$ \cite{ames2019control} is introduced in the final QP formulation in place of the nominal controller with the form as follows:
\begin{equation} \label{eq:clf-equation}
V_{clf}(x) = (z_1-z^d_1)^2 + (z_2-z^d_2)^2 + k(\theta - \theta^d)^2,
\end{equation}
where $\theta_d = -\frac{\pi}{4}$.
The desired position $(z^d_1,z^d_2)$ is chosen to be at the end of the corridor, and the initial orientation is $\frac{\pi}{4}$.
The CLF constraint $\dot{V}_{clf} \geq - \alpha_1 V_{clf} - s$ is then enforced, where $s \geq 0$ is a slack variable.
The slack variable ensures that the feasibility of the QP-based program is not affected by the CLF constraint.
The margin $\epsilon_1$ as defined in Rem. \ref{rem:eps-1} is chosen as $\SI{1.5}{cm}$ and $\epsilon_2$ is chosen as $10^{-5}$.

%%%%%%%%%%%%%%%%%%%%%%%%%%%%%%%%%%%%%%%%%%%%%%%%%%%%%%%%%%%%%%%%%%%

\subsection{Results}

The simulations are performed on a Virtual Machine with $4$ cores of a $\SI{2.20}{Ghz}$ Intel Core i7 processor, running IPOPT \cite{wachter2006implementation} on MATLAB, and the visualization is generated by MPT3~\cite{MPT3}.
The snapshots of the sofa trajectory is shown in Fig. \ref{fig:snapshots-sofa-problem} and also illustrated in the multimedia attachment.

\subsubsection{Enforcement of NCBF constraint}
The system remains safe throughout the simulation since the NCBF $h^{ij}$ is greater than $\epsilon_1^2$ for all possible robot-obstacle interactions as depicted in Fig. \ref{fig:min-d}.
In Fig. \ref{fig:cbf-cons}, the values of $\dot{h}(t)$, and LHS and RHS terms in the NCBF constraints from \eqref{eq:dual-QP} are shown. We can see that by Lem. \ref{lem:deriv-h-dual-QP}, the LHS term in the NCBF constraint is always a lower bound to $\dot{h}(t)$. This verifies the safety property of our duality-based approach.

\begin{figure}
    \centering
    % \begin{subfigure}[t]{0.99\linewidth}
    %     \centering
    %     \includegraphics[width=1\linewidth]{figures/sofa-h-1_small.png}
    %     \caption{NCBF (Square of minimum distance) for arm 1 of sofa with the walls}
    %     % \label{fig:min-d-obj-1}
    % \end{subfigure}
    % \begin{subfigure}[t]{0.99\linewidth}
    %     \centering
    %     \includegraphics[width=1\linewidth]{figures/sofa-h-2_small.png}
    %     \caption{NCBF (Square of minimum distance) for arm 2 of sofa with the walls}
    %     % \label{fig:min-d-obj-2}
    % \end{subfigure}
    \includegraphics[width=0.99\linewidth]{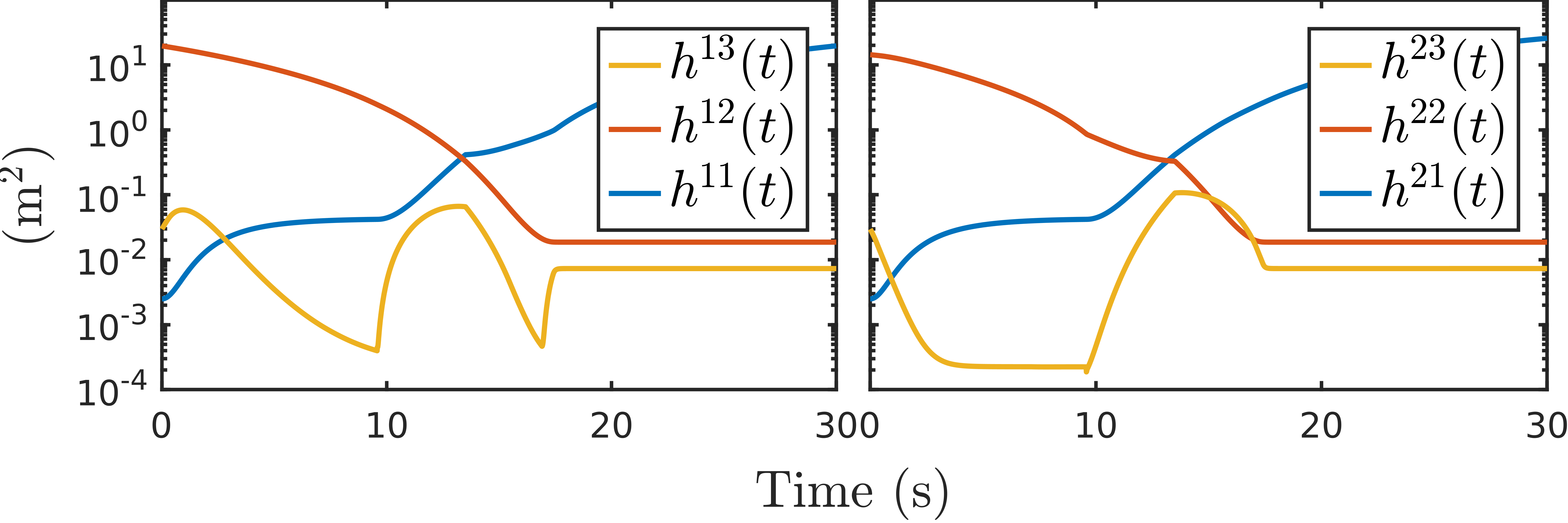}
    \caption{Square of minimum distance (NCBF) between the arms of the sofa and the walls, where Wall $W^1$ (left figure) is the left wall, Wall $W^2$ (right figure) the upper one, and Wall $W^3$ the inner wall of the corridor. Since the minimum distance is always greater than zero, the L-shaped sofa never collides with the obstacles over its trajectory. Both plots have log-scale on the y-axis.}
    \label{fig:min-d}
    \vspace{-10pt}
\end{figure}

\subsubsection{Computation time}

\begin{table}[h]
    \caption{Statistical analysis of computation time (ms) per iteration}
    \label{table:comp-time}
    \centering
    \begin{tabular}{|p{0.33 \linewidth}|>{\centering\arraybackslash}p{0.17\linewidth}|>{\centering\arraybackslash}p{0.08\linewidth}|>{\centering\arraybackslash}p{0.08\linewidth}|>{\centering\arraybackslash}p{0.08\linewidth}|}
    \hline
    Timing ($\SI{}{ms}$) & mean $\pm$ std  & p50 & p99 & max \\\hline
    Distance QPs \eqref{eq:min-dist-dual} & 1.07 $\pm$ 0.29 & 0.99 & 1.94 & 20.1 \\\hline
    Polytope-NCBF-QP \eqref{eq:dual-QP} & 14.5 $\pm$ 1.55 & 14.2 & 19.3 & 37.6 \\\hline
    Total ($2 \times 3 + 1 = 7$ QPs) & 21.1 $\pm$ 1.69 & 20.8 & 26.9 & 41.6 \\\hline
    \end{tabular}
\end{table}

From Table \ref{table:comp-time}, we can see that there are large outliers in computation time per iteration, but they occur in less than $0.1\%$ iterations.
% As $h(t)$ approaches zero, both sides of the NCBF constraint approach zero and more iterations are needed for the solver to find feasible solutions.
Nevertheless, from Table \ref{table:comp-time}, we can apply our controller at 50Hz.
The fast computation time allows us to directly implement the control inputs from \eqref{eq:dual-QP} on the robot in real-time. % , i.e. \eqref{eq:dual-QP} can be solved and implemented online.

% We have also observed that distances obtained as the optimal value of a QP can become negative due to numerical errors.
% In this case, the QP might become infeasible.

The optimization problem for the sofa problem has 51 variables and at most 72 constraints.
In general consider $N$ controlled robots in a $d$-dimensional space with $f$ facets and $m$ control inputs.
Then, the polytope-NCBF-QP \eqref{eq:dual-QP} has $(1{+}d)\frac{N(N-1)}{2}$ constraints along with at most $2f\frac{N(N-1)}{2}$ non-negativity constraints and $N(m{+}2f)$ variables, whereas a CBF-QP formulation using spherical over-approximation would have $\frac{N(N-1)}{2}$ constraints and $Nm$ variables.

\subsubsection{Continuity of \texorpdfstring{$\dot{h}(t)$}{dh/dt(t)} and \texorpdfstring{$\lambda^{*}(t)$}{lambda(t)}}
Fig. \ref{fig:cbf-cons} also shows that both $\dot{h}(t)$ and the lower bound of $\dot{h}(t)$ can have discontinuities.
For the case of the moving sofa problem, a discontinuity in $\dot{h}(t)$ can arise when the sofa is rotating and the point of minimum distance on the sofa with any wall jumps from one vertex to another, since the end points of the sofa arm need not have the same velocity when rotating.
The dual optimal variables $\lambda^{*i}(t)$ and $\lambda^{*j}(t)$ are always continuous and right-differentiable as shown in Sec. \ref{sec:dual-formulation}.
The dual variables are plotted in Fig. \ref{fig:dual-soln}, which demonstrate this property.

\subsubsection{Deadlocks}
For various initial conditions, the sofa can get stuck in a deadlock in the corridor.
This can happen when the arms of the sofa are so large that it cannot turn at the corner.
It can also happen when the two arms of the sofa get too close to the wall and the sofa cannot turn because it would cause one of its arms to penetrate the wall.
Our controller still ensures safety in this case.
A high-level planner could help by generating a deadlock free trajectory at low frequency that then serves as input to our control law.

%%%%%%%%%%%%%%%%%%%%%%%%%%%%%%%%%%%%%%%%%%%%%%%%%%%%%%%%%%%%%%%%%%%

\begin{figure}
    \centering
    % \begin{subfigure}[t]{0.99\linewidth}
    %     \centering
    %     % \includegraphics[width=1\linewidth]{figures/sofa-dhdt-11.eps}
    %     \includegraphics[width=1\linewidth]{figures/sofa-dhdt-11_small.png}
    %     \caption{NCBF constraint enforcement between Wall 1 and Arm 1. The wall is very far away from Arm 1 after $\SI{20}{s}$ and is cropped from the plot.}
    %     \label{fig:cbf-cons-wall-1}
    % \end{subfigure}
    % \begin{subfigure}[t]{0.99\linewidth}
    %     \centering
    %     % \includegraphics[width=1\linewidth]{figures/sofa-dhdt-13.eps}
    %     \includegraphics[width=1\linewidth]{figures/sofa-dhdt-13_small.png}
    %     \caption{NCBF constraint enforcement between Wall 3 and Arm 1}
    %     \label{fig:cbf-cons-wall-3}
    % \end{subfigure}
    \includegraphics[width=0.99\linewidth]{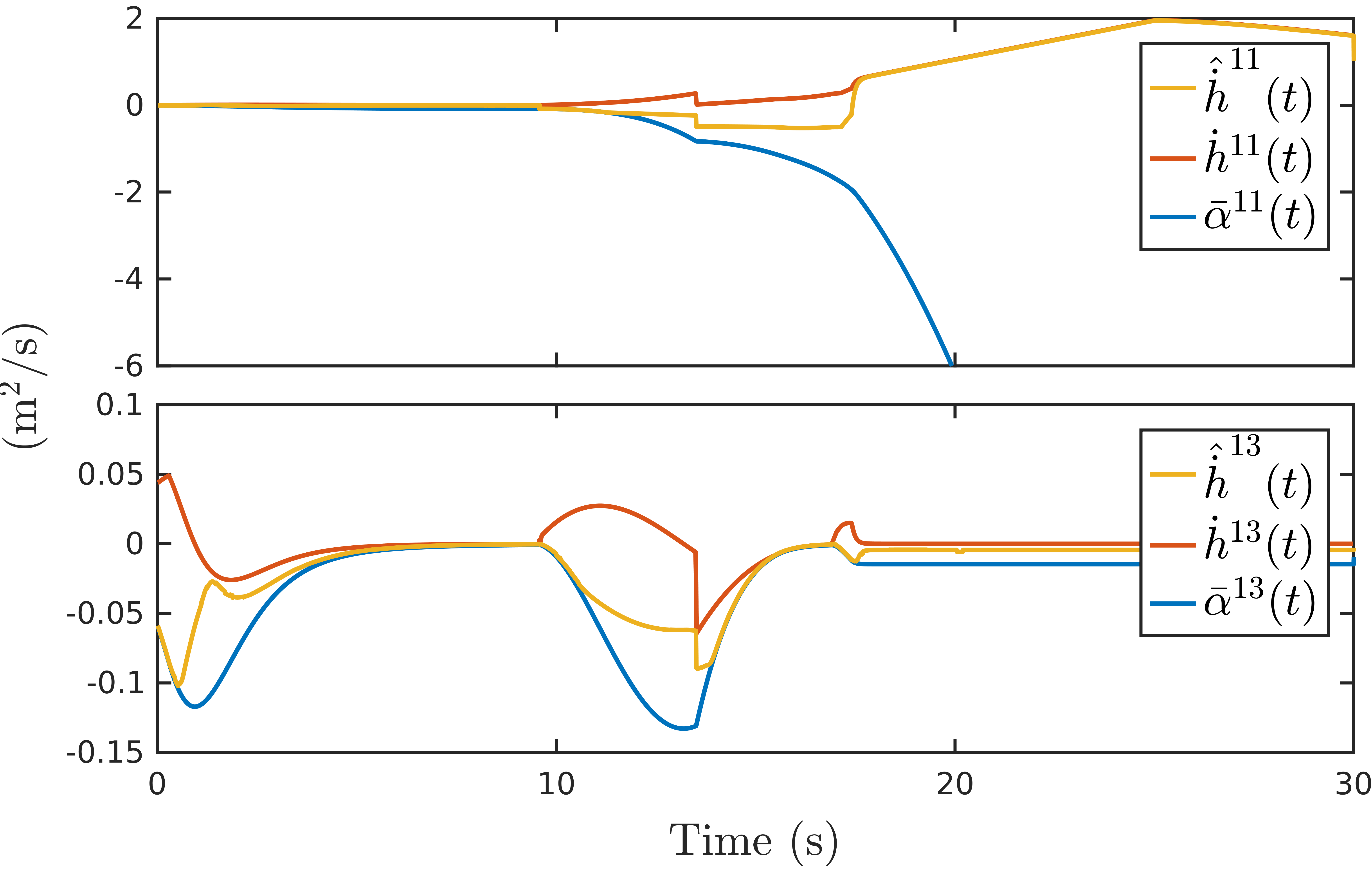}
    \caption{NCBF constraint enforcement between Arm 1 and Wall 1 (top figure) and Arm 1 and Wall 3 (bottom figure).
    Illustration of safety performance of the system. The red lines are $\dot{h}(t)$, the yellow lines are the lower bounds of $\dot{h}(t)$, and blue line is the RHS of the NCBF constraint. $\bar{\alpha}^{1j}(t) = -\alpha(h^{1j}(t)-\epsilon_1^2)$, where $j$ is the wall.
    The system safety is guaranteed with the red line being always above the blue line.}
    \label{fig:cbf-cons}
    \vspace{-10pt}
\end{figure}

\subsection{Discussions}

\subsubsection{Nonlinearity of the system dynamics}
The duality-based formulation \eqref{eq:dual-QP} is a convex quadratic program even when the system dynamics is nonlinear, as long as it is control affine.
This allows us to achieve dynamically-feasible obstacle avoidance with QPs for polytopes.

\subsubsection{Optimality of solution vs computation time}
As noted in Sec. \ref{sec:dual-formulation}, the cost function of the QP-based program \eqref{eq:dual-QP} does not affect the safety of the system.
If the optimization solver does not converge to the optimal solution, the current solution can be used if it is feasible.
This can be useful in real-time implementations where both control frequency and safety matter.
A feasible solution to \eqref{eq:dual-QP} can be directly found in some cases, such as when the NCBF is constructed using a safe backup controller \cite{squires2018constructive}.

\subsubsection{Trade-off between computation speed and tight maneuvering}
A polytope with $f^i$ faces would require $f^i$ dual variables in the dual formulation \eqref{eq:dual-QP} and additional constraints.
Using a hyper-sphere as an over-approximation to the polytope would require fewer dual variables, as in \cite{xiao2019control}.
However, such an approximation can be too conservative and completely ignore the rotational geometry of the polytope.
Using the full polytope structure can prove beneficial in dense environments, such as the sofa problem, where a spherical approximation cannot work.
So, there is a trade-off between computation speed and maneuverability.
In practice, a hybrid approach should be used: a hyper-sphere approximation when two obstacles are far away, and the polytope structure when closer, which could find a good trade-off between computation speed and maneuverability for tight obstacle avoidance in dense environments.

% ADD THEM IN THE JOURNAL PAPER
% \subsubsection{Computation time vs. number of faces of a polytope}
% As the number of faces of a polytope increases, so does the number of dual variables.
% For the duality-based formulation \eqref{eq:dual-QP}, as the number of faces of any polytope increase, both the number of constraints and variables of \eqref{eq:dual-QP} increase linearly, and thus the computation time also increases.

\subsubsection{Robustness with respect to the safe set}
Since we use the minimum distance between two polytopes as the NCBF and not the signed distance, the minimum distance is uniformly zero when two polytopes intersect.
So, the proposed controller is not robust in the sense that if the state leaves the safe set, it will not converge back to it.

% ADD THEM IN THE JOURNAL PAPER
% \subsubsection{Numerical errors}
% If the controller is implemented in a discretized first-order hold form and the time step is too large, the obstacles can intersect each other.
% Collisions can occur due to large velocities and discontinuous $\dot{h}(t)$.
% The maximum velocity of any point on an obstacle can be calculated using the velocity and angular velocity bounds, and the diameter of the corresponding polytopic domain.
% The maximum velocity should be such that the maximum possible distance any point on the robot can cover in a time step is less than the $\epsilon_1$ margin.
% $\epsilon_1$ should be judiciously chosen so that the system is safe but, not too conservative.
% This arises due to the discretization of the problem and using a fixed-time integration scheme.
% These issues can be prevented by considering the discrete-time version of this problem directly.
\section{Conclusion}
\label{sec:conclusion}
In this paper, we have presented a general framework for obstacle avoidance between polytopes using the dual program of a minimum distance QP problem.
We have shown that the control input using our method can be computed using a QP for systems with control affine dynamics, enabling real-time implementation.
We have numerically verified the safety performance of our controller for the problem of moving an L-shaped sofa through a tight corridor.
% In the future, we will generalize our framework and provide proofs for the assumptions.
We also have explored properties such as robustness and deadlock avoidance.
Further work in this topic would involve extending our method to other convex shapes, and deriving safety guarantees for more general classes of systems.

\begin{figure}
    \centering
    % \begin{subfigure}[t]{0.99\linewidth}
    %     \centering
    %     \includegraphics[width=1\linewidth]{figures/dual-soln-13_small.png}
    %     \caption{The dual optimal solutions between arm 1 of the robot and wall 3. $\lambda^{W^3}(t)$ and $\lambda^{A^1}(t)$ represent the dual optimal solutions corresponding to wall 3 and arm 1 respectively.}
    %     \label{fig:dual-soln-1}
    % \end{subfigure}
    % \begin{subfigure}[t]{0.99\linewidth}
    %     \centering
    %     \includegraphics[width=1\linewidth]{figures/dual-soln-23_small.png}
    %     \caption{The dual optimal solutions between arm 2 of the robot and wall 3. $\lambda^{W^3}(t)$ and $\lambda^{A^2}(t)$ represent the dual optimal solutions corresponding to wall 3 and arm 2 respectively.}
    %     \label{fig:dual-soln-2}
    % \end{subfigure}
    \includegraphics[width=0.99\linewidth]{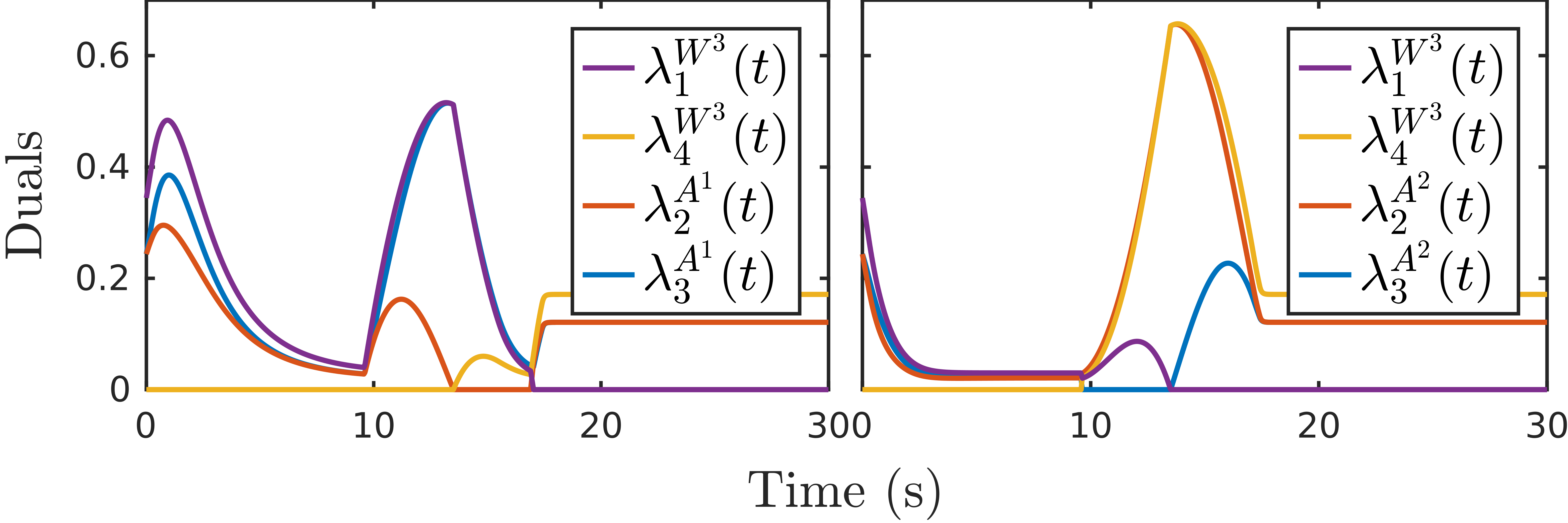}
    \caption{Dual optimal solutions for Arm 1 (left figure) and Arm 2 (right figure) of the robot corresponding to the obstacle wall 3. The dual optimal solutions are unique, continuous, and right-differentiable as shown in Sec. \ref{sec:dual-formulation}. Note that pair of dual variables for Arm 1-Wall 3 are different from that of Arm 2-Wall 3, i.e. the vector $\lambda^{W^3}(t)$ is different for Arms 1 and 2. The subscripts represent the components of the vectors, and only the non-zero components have been plotted.}
    \label{fig:dual-soln}
    \vspace{-10pt}
\end{figure}

\balance

\section*{Acknowledgement}
\label{sec:acknowledgement}
The authors gratefully acknowledge Somayeh Sojoudi at the University of California, Berkeley for her valuable comments on properties of parametric optimization problems.

\bibliographystyle{IEEEtran}
% \balance
\bibliography{bio/bibliography}
\appendix

\subsection{Proof of Lem. \ref{lem:h-dual-qp}}
\label{app:proof-h-dual-qp}

The Lagrangian function for \eqref{eq:min-dist-primal} is~\cite[Chap.~5]{boyd2004convex}:
\begin{align} \label{eq:lagrange-func}
\Lambda(z^i, z^j, \lambda^i, \lambda^j) = & \lambda^j(A^j(x^j)z^j - b^j(x^j)) + \\\nonumber
& \lambda^i(A^i(x^i)z^i -b^i(x^i)) + \lVert z^i - z^j \rVert^2,
\end{align}
where $\lambda^i \in \R^{1\times r^i}, \lambda^j \in \R^{1\times r^j}$ are the dual variables.
The Lagrangian dual function can be written as:
\begin{equation} \label{eq:lagrange-dual-func}
L(\lambda^i, \lambda^j) = \inf_{z^i,z^j \in \R^l} \left(\Lambda(z^i,z^j,\lambda^i,\lambda^j)\right),
\end{equation}
and together with the Weak Duality Theorem~\cite[Chap.~5]{boyd2004convex}, we have,
\begin{equation} \label{eq:WDT}
L(\lambda^i, \lambda^j) \leq h(x), \;\; \forall \; \lambda^i, \lambda^j \geq 0.
\end{equation}

As the constraints in \eqref{eq:min-dist-primal} are affine in $z^i$ and $z^j$ at any given time, constraint qualification holds and the Strong Duality Theorem can be applied, resulting in
\begin{equation} \label{eq:SDT}
\max_{\lambda^i, \lambda^j \geq 0} L(\lambda^i,\lambda^j) = h(x).
\end{equation}
$L(\lambda^i,\lambda^j)$ can be explicitly computed. Let the minimizer for \eqref{eq:lagrange-dual-func} be $z^i, z^j$. Then,
\begin{align}
\label{eq:minimizer_of_L}
\dfrac{\partial \Lambda}{\partial z^i} = 0 & \Rightarrow 2(z^i-z^j)^T + \lambda^iA^i(x^i) = 0,\\\nonumber
\dfrac{\partial \Lambda}{\partial z^j} = 0 & \Rightarrow 2(z^j-z^i)^T + \lambda^jA^j(x^j) = 0.
\end{align}
Hence, for the minimum of \eqref{eq:lagrange-dual-func} to exist, from \eqref{eq:minimizer_of_L} the following must hold:
\begin{align} \label{eq:lagrange-inf-cond}
\lambda^iA^i(x^i) + \lambda^jA^j(x^j) = 0,\\\nonumber
\lambda^iA^i(x^i) = -2(z^i-z^j)^T.
\end{align}
Substituting the value of $(z^i-z^j)$ from \eqref{eq:lagrange-inf-cond} in \eqref{eq:lagrange-func} and using \eqref{eq:lagrange-dual-func}, we have
\begin{equation} \label{eq:lagrange-dual-func-expl-proof}
\begin{split}
    L(\lambda^i, \lambda^j) = & -\frac{1}{4}\lambda^iA^i(x^i)A^i(x^i)^T\lambda^{i \; T} - \lambda^ib^i(x^i)\\
    & - \lambda^jb^j(x^j).
\end{split}
\end{equation}

The dual formulation \eqref{eq:dual-QP} then follows from \eqref{eq:lagrange-inf-cond} and \eqref{eq:lagrange-dual-func-expl-proof}.

\subsection{Proof of Theorem \ref{theorem:safety}}
\label{app:safety-proof}

Let the feasible set of control inputs for \eqref{eq:dual-QP} be $\mathcal{F}_u(x)$.
The control input $u^*(x)$ is feasible for \eqref{eq:dual-QP}, and hence satisfies \eqref{subeq:dual-QP-ncbf}, which implies safety by \eqref{eq:strong-nbf-proof}.
However, it shall be noted that since $u^*$ need not be continuous, the Filippov operator \eqref{eq:filippov-operator-def} could be required to obtain a valid closed loop trajectory.
The control inputs obtained from the Filippov operation might not be feasible for \eqref{eq:dual-QP}, and we might lose the safety property.

The overview of the proof is as follows:
We show that: (1) for any control law chosen from the feasible set of inputs $\mathcal{F}_u(x)$, the control input obtained after applying the Filippov operator is still feasible for \eqref{eq:dual-QP}, and (2) any control input feasible for \eqref{eq:dual-QP} results in a safe trajectory.

Therefore, we first need to show that inputs obtained from the Filippov operator $F$ is feasible for \eqref{eq:dual-QP}.
We can prove this by showing $F[\mathcal{F}_u](x) = \mathcal{F}_u(x)$.

\subsubsection{\texorpdfstring{$F[\mathcal{F}_u](x) = \mathcal{F}_u(x)$}{F[Fu](x)=Fu(x)}}
\label{subsec:app-safety-proof-usc}

Recall, from the definition of the Filippov operator \eqref{eq:filippov-operator-def}, that the Filippov operator applied on a set-valued function makes it closed and convex point-wise and upper semi-continuous.
To show that $F[\mathcal{F}_u](x) = \mathcal{F}_u(x)$, we can equivalently show that $\mathcal{F}_u$ is closed and convex point-wise, and upper semi-continuous, thus making it invariant to the Filippov operator.

By assumption, \eqref{eq:dual-QP} is always feasible.
So, $\mathcal{F}_u(x) \neq \emptyset \; \forall \; x \in \mathcal{S}$.
Further, since $\mathcal{U}$ is convex and compact, $\mathcal{F}_u(x)$ is convex and compact pointwise $\forall \; x \in \mathcal{S}$.
Finally,  we need to show that $\mathcal{F}_u(x)$ is upper semi-continuous.
Consider any sequences $\{x_{(p)}\} \rightarrow x$ and $\{u_{(p)}\} \rightarrow u$, where $u_{(p)}$ are feasible for \eqref{eq:dual-QP} at $x_{(p)}$, i.e. $u_{(p)} \in \mathcal{F}_u(x_{(p)})$.
To show that $\mathcal{F}_u$ is upper semi-continuous, we have to show that $u$ is feasible for \eqref{eq:dual-QP} at $x$, i.e. $u \in \mathcal{F}_u(x)$.
This is not trivial to show because the number of constraints in \eqref{subeq:dual-QP-non-neg} changes depending on $x$.

We now prove some general properties for optimization problems, using which we can conclude upper semi-continuity of $\mathcal{F}_u$.
Consider an optimization problem $O$ of the form $\rho^*(\xi) = \max_v \{\rho(\xi,v): v \in \Pi(\xi,v)\}$, where $\Pi(\xi,v) = \{v: \pi_k(\xi,v) \leq 0, k \in K(\xi)\}$, where $\rho$ is the continuous cost function, $\pi$ is a continuous vector function, $K$ is an index set, and the inequality is applied element-wise.
Let $\Pi(\xi,v)$ be non-empty and uniformly bounded in $\xi$, and let $v^*(\xi)$ be an optimal solution at $\xi$.
Note that the constraints enforced in the optimization can vary with $\xi$ since the index set $K$ depends on $\xi$.
Here are the three properties as follows,
\begin{enumerate}[label=(\alph*)]
    \item \label{prop:app-safety-index-subset} If the index set is reduced, the optimal cost increases:
    
    Define $\bar{\Pi}(\xi,v) = \{v: \pi_k(\xi,v) \leq 0, k \in \bar{K}\}$, where $\bar{K} \subseteq K(x) \; \forall \xi$.
    Define the corresponding optimization problem $\bar{O}$ as $\bar{\rho}^*(\xi) = \max_v \{\rho(\xi,v): v \in \bar{\Pi}(\xi,v)\}$, with optimal solution $\bar{v}^*(\xi)$.
    Since $\Pi(\xi, v) \subseteq \bar{\Pi}(\xi,v)$, $\bar{\rho}^*(\xi) \geq \rho^*(\xi) \; \forall \xi$.
    
    \item \label{prop:app-safety-continuity} Continuity property of optimal cost $\bar{\rho}^*$:
    
    Consider a sequence $\{\xi_{(p)}\} {\rightarrow} \xi$ and optimal solutions $\{\bar{v}^*(\xi_{(p)})\} {\rightarrow} \bar{v}$.
    By continuity of $\pi$ and $\rho$, $\pi_{\bar{K}}(\xi,\bar{v}) = \lim_{p\rightarrow \infty} \pi_{\bar{K}}(\xi_{(p)}, \bar{v}^*(\xi_{(p)})) \leq 0$, i.e. $(\xi,\bar{v})$ is a feasible solution for $\bar{O}$, and $ \bar{\rho}^*(\xi) \geq \rho(\xi,\bar{v}) = \lim_{p\rightarrow \infty} \rho(\xi_{(p)}, \bar{v}^*(\xi_{(p)})) = \lim_{p\rightarrow \infty} \bar{\rho}^*(\xi_{(p)})$.
    
    \item \label{prop:app-safety-limit-index} Limit property of the index set, $K(\xi)$:
    
    Let $K(\xi)$ be of the form $K(\xi) = \{k: \eta_k(\xi) < \epsilon\}$, where $\eta$ is a continuous, non-negative vector function and $\epsilon > 0$.
    Consider a sequence $\{\xi_{(p)}\} {\rightarrow} \xi$.
    By continuity of $\eta$, $\exists P \in \N$ such that $\forall p \geq P$, and $\forall k \in K(\xi)$, $\eta_k(\xi_{(p)}) < \epsilon$.
    Thus, $K(\xi) \subseteq K(\xi_{(p)}) \; \forall p \geq P$.
\end{enumerate}

We now use the above properties to show upper semi-continuity of $\mathcal{F}_u$, which is the feasible set of inputs for \eqref{eq:dual-QP}.
The variables $\xi$ and $v$ in the optimization problem $O$ correspond to the pair $(x, u)$ and $\dot{\lambda}$ in \eqref{eq:dual-QP} respectively.
Define the index set as $K(y) = \{(k^i, k^j): \lambda^{*i}_k(y) < \epsilon_2, \lambda^{*j}_k(y) < \epsilon_2 \} \subseteq [r^i] \times [r^j]$ corresponding to \eqref{subeq:dual-QP-non-neg}.
Note that $\lambda^{*i}$ and $\lambda^{*j}$, corresponding to $\eta$, are continuous functions (Sec. \ref{subsec:ncbf-dual-formulation}).

% By Lem.~\ref{lem:differentiable-separating-vector} and LICQ, $(\lambda^{*i},\lambda^{*j})$ is continuous at $x$.
% Let $K^i(y) \subset [r^i], K^j(y) \subset [r^j]$ be such that $\lambda^{*i}(y)_{k} < \epsilon_2 \; \forall \; k \in K^i(y)$ and $\lambda^{*j}(y)_{k} < \epsilon_2 \; \forall \; k \in K^j(y)$.
% By continuity of $(\lambda^{*i},\lambda^{*j})$, $\exists \; P \in \N$ such that $\forall \; p \geq P$, $K^i(x) \subset K^i(x_{(p)})$ and $K^j(x) \subset K^j(x_{(p)})$.

Now consider the following optimization problem, denoted as LP$(x,y,u)$ relating $O$ to \eqref{eq:dual-QP}:
\begin{subequations} \label{eq:app-safety-g-def}
\begin{align}
  g(x,y,u) {=} & \max_{\{\dot{\lambda}^i,\dot{\lambda}^j\}} \dot{L}(t,\lambda^{*i}(x),\lambda^{*j}(x),\dot{\lambda}^i,\dot{\lambda}^j,u) \label{subeq:app-safety-g-def-cost} \\
    \text{s.t.} \quad & \dot{\lambda}^iA^i(x) {+} \lambda^{*i}(x)(\mathcal{L}_{f^i}A^i(x) {+} \mathcal{L}_{g^i}A^i(x)u) \hspace{-1pt} \label{subeq:app-safety-g-def-plane} \\\nonumber
    & \quad {+} \lambda^{*j}(x)(\mathcal{L}_{f^j}A^j(x) {+} \mathcal{L}_{g^j}A^j(x)u) {=} {-}\dot{\lambda}^jA^j(x) \\
%    & \quad = 0, \\
    & \dot{\lambda}^i_k  \geq 0 \quad \text{if } k \in K^i(y), \label{subeq:app-safety-g-def-non-neg} \\\nonumber
    & \dot{\lambda}^j_k  \geq 0 \quad \text{if } k \in K^j(y), \\
    & |\dot{\lambda}^{i}| \leq M, |\dot{\lambda}^{j}| \leq M. \label{subeq:app-safety-g-def-bounds}
\end{align}
\end{subequations}
Note that the cost \eqref{subeq:app-safety-g-def-cost} of LP$(x,y,u)$ corresponds to the LHS of \eqref{subeq:dual-QP-ncbf}, and the constraints \eqref{subeq:app-safety-g-def-plane}-\eqref{subeq:app-safety-g-def-bounds} correspond to \eqref{subeq:dual-QP-plane}-\eqref{subeq:dual-QP-bounds}.
The parameter $y$ in LP$(x,y,u)$ only affects the index set in \eqref{subeq:app-safety-g-def-non-neg}.
Thus, $u$ if feasible at $x$ only if $g(x,x,u) \geq -\alpha(h(x)-\epsilon_1^2)$.
Also, since $u_{(p)}$ is feasible at $x_{(p)}$, $g(x_{(p)},x_{(p)},u_{(p)}) \geq -\alpha(h(x_{(p)})-\epsilon_1^2)$.

The optimization LP$(x,x,u)$ corresponds to the optimization problem $O$ as defined above.
Note that by continuity of $A^i$, $A^j$, $b^i$, $b^j$, $\lambda^{*i}$, and $\lambda^{*j}$, the cost \eqref{subeq:app-safety-g-def-cost} and constraint matrices \eqref{subeq:app-safety-g-def-plane}-\eqref{subeq:app-safety-g-def-bounds} are continuous, similar to $O$.
Moreover, the feasible set of \eqref{eq:app-safety-g-def} is non-empty (by the assumption in Thm. \ref{theorem:safety}) and is uniformly bounded due to \eqref{subeq:app-safety-g-def-bounds}.
Thus, the optimization problem LP$(x,x,u)$ fits the framework of $O$.

By Property~\ref{prop:app-safety-limit-index}, $\exists P {\in} \N$ such that $K(x) \subseteq K(x_{(p)}) \; \forall p \geq P$.
Truncate the sequences $\{x_{(p)}\}$ and $\{u_{(p)}\}$ before $P$.
Define the optimization problem LP$(x_{(p)}, x, u_{(p)})$ corresponding to $\bar{O}$.
By Property~\ref{prop:app-safety-index-subset}, $g(x_{(p)}, x, u_{(p)}) \geq g(x_{(p)}, x_{(p)}, u_{(p)}) \; \forall p$.
Let the optimal solution to LP$(x_{(p)}, x_{(p)}, u_{(p)})$ be $\dot{\lambda}_{(p)}$.
Since $\dot{\lambda}_{(p)}$ is uniformly bounded (by \eqref{subeq:app-safety-g-def-bounds}), by Bolzano-Weierstrass Theorem, there exists a sub-sequence $\dot{\lambda}_{(q)}$ such that it converges to $\dot{\lambda}$.
We now restrict the analysis to the $q$ sequence.
By Property~\ref{prop:app-safety-continuity} for $\bar{O}$, we have $g(x,x,u) \geq \lim_{q \rightarrow \infty} g(x_{(q)},x,u_{(q)})$.

Combining all the inequalities from above, and by continuity of $h(x)$ (Lem. \ref{lem:differentiable-separating-vector}),
\begin{align*}
    g(x,x,u) & \geq \lim_{q \rightarrow \infty} g(x_{(q)},x,u_{(q)}) \\
    & \geq \lim \sup_{q \rightarrow \infty} g(x_{(q)},x_{(q)},u_{(q)}) \\
    & \geq \lim \sup_{q \rightarrow \infty} -\alpha(h(x_{(q)})-\epsilon_1^2) = -\alpha(h(x)-\epsilon_1^2).
\end{align*}
Thus, $u$ is feasible at $x$ and $\mathcal{F}_u(x)$ is upper-semicontinuous.

\subsubsection{Safety}
\label{subsec:app-safety-proof-safety}

Finally, we show that the Filippov operator does not affect the safety of the system.
Let $u^*(x)=(u^{*i}(x),u^{*j}(x)) \in \mathcal{F}_u(x)$ be any measurable feedback control law, obtained as the solution of \eqref{eq:dual-QP}.
Then, taking the Filippov operator,
\begin{equation*}
F[u^{*}](x) \subset F[\mathcal{F}_u](x) = \mathcal{F}_u(x).
\end{equation*}
Given $x(0)$ such that $h(0) > \epsilon_1^2$, a closed loop trajectory $(x^i(t),x^j(t)), t \in [0,T]$ can be obtained using the Filippov solution, where $(x^i(t),x^j(t))$ are absolutely continuous in $[0,T]$ and for almost all $t \in [0,T]$, satisfy,
\begin{align*}
\dot{x}^i(t) & = f^i(x(t)) + g^i(x(t))u^i(t) \\
\dot{x}^j(t) & = f^j(x(t)) + g^j(x(t))u^j(t),
\end{align*}
where $u(t) = (u^i(t),u^j(t)) \in F[u^{*}](x(t)) \subset \mathcal{F}_u(x(t))$.
Note that the constraints of \eqref{eq:dual-QP} are stronger than those of \eqref{eq:deriv-h-dual-QP}, meaning that for any $(u,\dot{\lambda}^i, \dot{\lambda}^j)$ feasible for \eqref{eq:dual-QP}, $(\dot{\lambda}^i, \dot{\lambda}^j)$ is feasible for \eqref{eq:deriv-h-dual-QP}, and by \eqref{subeq:dual-QP-ncbf} and Lem.~\ref{lem:deriv-h-dual-QP}, $(\dot{\lambda}^i, \dot{\lambda}^j)$ satisfies \eqref{eq:strong-nbf-proof}.
Since $u(t)$ is feasible for \eqref{eq:dual-QP}, $\exists \; (\dot{\lambda}^i(t),\dot{\lambda}^j(t))$ such that the closed loop trajectory satisfies \eqref{subeq:dual-QP-ncbf}:
\begin{equation*}
\dot{L}(t,\lambda^{*i}(t),\lambda^{*j}(t),\dot{\lambda}^i(t),\dot{\lambda}^j(t)) \geq  -\bar{\alpha}(h(t)-\epsilon_1^2),
\end{equation*}
for almost all $t \in [0,T]$.
By \eqref{eq:strong-nbf-proof}, $\dot{h}(t) \geq -\alpha(h(t)-\epsilon_1^2)$ for almost all $t \in [0,T]$.
By Lem. \ref{lem:nbf-def}, $h(t) > \epsilon_1^2 \; \forall \; t \in [0,T]$, and thus the system remains safe.

\end{document}